\documentclass[conference, 10pt]{IEEEtran}
\IEEEoverridecommandlockouts
\usepackage{cite}
\usepackage[numbers,sort&compress]{natbib}
\usepackage{url}
\usepackage{amsmath,amssymb,amsfonts}
\usepackage{algorithmic}
\usepackage{graphicx}
\usepackage{caption, multirow}
\usepackage{textcomp}
\usepackage{xcolor}
\def\BibTeX{{\rm B\kern-.05em{\sc i\kern-.025em b}\kern-.08em
    T\kern-.1667em\lower.7ex\hbox{E}\kern-.125emX}}
\usepackage{amsthm}
\theoremstyle{definition}
\newtheorem{definition}{Definition}
\newtheorem{proposition}{Proposition}
\usepackage{hyperref}

\renewcommand{\baselinestretch}{1}
    
\begin{document}

\title{Co-clustering Vertices and Hyperedges via\\Spectral Hypergraph Partitioning
\thanks{
This work was supported by NSF under award CCF-2008555.
B. Li was partially supported by the Ken Kennedy Institute 2020/21 Ken Kennedy-Cray Graduate Fellowship.
We also acknowledge the support of NVIDIA Corporation.
E-mails: \{yz126, boning.li, segarra\}@rice.edu }}

\author{
Yu Zhu, Boning Li, Santiago Segarra \\
\textit{Rice University, USA}
}

\maketitle

\begin{abstract}
We propose a novel method to co-cluster the vertices and hyperedges of hypergraphs with edge-dependent vertex weights (EDVWs).
In this hypergraph model, the contribution of every vertex to each of its incident hyperedges is represented through an edge-dependent weight, conferring the model higher expressivity than the classical hypergraph.
In our method, we leverage random walks with EDVWs to construct a hypergraph Laplacian and use its spectral properties to embed vertices and hyperedges in a common space.
We then cluster these embeddings to obtain our proposed co-clustering method, of particular relevance in applications requiring the simultaneous clustering of data entities and features. 
Numerical experiments using real-world data demonstrate the effectiveness of our proposed approach in comparison with state-of-the-art alternatives.
\end{abstract}



\section{Introduction}\label{sec-intro}

Clustering, a fundamental task in data mining and machine learning, aims to divide a set of entities into several groups such that entities in the same group are more similar to each other than to those in other groups. 
In \emph{graph} clustering or partitioning, the entities are modeled as the vertices of a graph and their similarities are encoded in the edges. 
In this setting, the goal is to group the vertices into clusters such that there are more edges within each cluster than across clusters. 

While graphs serve as a popular tool to model \emph{pairwise} relationships, in many real world applications the entities engage in more complicated, \emph{higher-order} relationships.
For example, in coauthorship networks~\cite{han2009understanding} more than two authors can interact in writing a manuscript. 
Hypergraphs can be used to represent such datasets, where the notion of an edge is extended to a hyperedge that can connect more than two vertices.
Existing research on hypergraph partitioning mainly follows two directions. 
One is to project a hypergraph onto a proxy graph via hyperedge expansion and then graph partitioning methods can be directly leveraged \cite{agarwal2005beyond, zhou2007learning, agarwal2006higher}. 
Another one is to represent hypergraphs using tensors and adopt tensor decomposition algorithms \cite{shashua2006multi, ghoshdastidar2015provable, chen2017fiedler, ke2019community}.

To better accommodate hypergraphs for the representation of real-world data, several extensions over the classical hypergraph have been recently proposed~\cite{li2017inhomogeneous, baytas2018heterogeneous, chitra2019, hayashi2020, schaub2021signal}. 
These more elaborate models consider different types of vertices or hyperedges, or different levels of relations. 
In this paper, we consider edge-dependent vertex weights (EDVWs)~\cite{chitra2019}, which can be used to reflect the different importance or contribution of vertices in a hyperedge.
This model is highly relevant in practice.
For example, an e-commerce system can be modeled as a hypergraph with EDVWs where users and products are respectively modeled as vertices and hyperedges, and EDVWs represent the quantity of a product in a user's shopping basket~\cite{li2018tail}. 
EDVWs can also be used to model the relevance of a word to a document in text mining~\cite{hayashi2020}, the probability of an image pixel belonging to a segment in image segmentation~\cite{ding2010interactive}, and the author positions in a coauthorship or citation network~\cite{chitra2019}, to name a few.

A large portion of clustering algorithms focus on one-way clustering, i.e., clustering data entities based on their features and, in the hypergraph setting, clustering vertices based on hyperedges. 
Indeed, in \cite{hayashi2020}, a hypergraph partitioning algorithm was proposed to cluster the vertices in a hypergraph with EDVWs. 
However, it is more desirable to simultaneously cluster (or co-cluster) both vertices and hyperedges in many applications including text mining~\cite{dhillon2001co, dhillon2003information}, product recommendation~\cite{vlachos2014improving}, and bioinformatics~\cite{cheng2000biclustering, cho2004minimum}. 
Moreover, co-clustering can leverage the benefit of exploiting the duality between data entities and features to effectively deal with high-dimensional and sparse data~\cite{dhillon2003information,long2005co}. 

In this paper, we study the problem of co-clustering vertices and hyperedges in a hypergraph with EDVWs. 
Our contributions can be summarized as follows:\\
(i)~We define a Laplacian for hypergraphs with EDVWs through random walks on vertices and hyperedges and show its equivalence to the Laplacian of a specific digraph obtained via a modified star expansion of the hypergraph.\\
(ii)~We propose a spectral hypergraph co-clustering method based on the proposed hypergraph Laplacian.\\
(iii)~We validate the effectiveness of the proposed method via numerical experiments on real-world datasets. 

Notation: 
The entries of a matrix $\mathbf{X}$ are denoted by $X_{ij}$ or $\mathbf{X}(i,j)$. 
Operations $(\cdot)^\top$ and $\mathrm{Tr}(\cdot)$ represent transpose and trace, respectively.
$\mathbf{1}$ and $\mathbf{I}$ refer to the all-ones vector and the identity matrix, where the sizes are clear from context.
$\mathbf{I}_{N}$ and $\mathbf{0}_{N\times M}$ refer to the identity matrix of size $N\times N$ and the all-zero matrix of size $N\times M$.
$\mathrm{diag}(\mathbf{x})$ denotes a diagonal matrix whose diagonal entries are given by the vector $\mathbf{x}$.
Finally, $[\mathbf{X};\mathbf{Y}]$ represents the matrix obtained by {vertically} concatenating two matrices $\mathbf{X}$ and $\mathbf{Y}$, while $[\mathbf{X}, \mathbf{Y}]$ denotes horizontal concatenation.


\section{Preliminaries}\label{sec-pre}

\subsection{Hypergraphs with edge-dependent vertex weights}\label{ssec-model}

Hypergraphs are generalizations of graphs where edges can connect more than two vertices.
In this paper, we consider the hypergraph model with EDVWs~\cite{chitra2019} as defined next.

\begin{definition}
A hypergraph $\mathcal{H}=(\mathcal{V},\mathcal{E},\omega,\gamma)$ with EDVWs consists of a set of vertices $\mathcal{V}$, a set of hyperedges $\mathcal{E}$ where a hyperedge is a subset of the vertex set, a weight $\omega(e)$ for every hyperedge $e\in\mathcal{E}$, and a weight $\gamma_e(v)$ for every hyperedge $e\in\mathcal{E}$ and every vertex $v\in e$. 
\end{definition}

The difference between the above hypergraph model and the typical hypergraph model considered in most existing papers is the introduction of the EDVWs $\{\gamma_e(v)\}$. 
The motivation is to enable the model to describe the cases when the vertices in the same hyperedge contribute differently to this hyperedge. 
For example, in a coauthorship network, every author (vertex) in general has a different degree of contribution to a paper (hyperedge), usually represented by the order of the authors.
This information is lost in traditional hypergraph models but it can be easily encoded through EDVWs.

For convenience, let $\mathbf{R}\in\mathbb{R}^{|\mathcal{E}|\times|\mathcal{V}|}$ collect edge-dependent vertex weights, with $R_{ev}=\gamma_e(v)$ if $v\in e$ and $0$ otherwise. 
Also, let $\mathbf{W}\in\mathbb{R}^{|\mathcal{V}|\times|\mathcal{E}|}$ collect hyperedge weights, with $W_{ve}=\omega(e)$ if $v\in e$ and $0$ otherwise. 
Throughout the paper we assume that the hypergraph is connected.

\subsection{Spectral graph partitioning}\label{ssec-sgp}

Given an undirected graph $\mathcal{G}$ with $N$ vertices, the goal of graph partitioning is to divide its vertex set into $k$ disjoint subsets (clusters) $\mathcal{S}_1,\cdots,\mathcal{S}_k$ such that there are more (heavily weighted) edges inside a cluster and few edges across clusters, while these clusters are also balanced in size.\footnote{Although there are different variations of the graph partitioning problem~\cite{bulucc2016recent}, this is the one that we adopt in this paper.}

To postulate this problem, let $\mathbf{A}_g$, $\mathbf{D}_g=\mathrm{diag}(\mathbf{A}_g\mathbf{1})$, and $\mathbf{L}_g=\mathbf{D}_g-\mathbf{A}_g$ denote the weighted adjacency matrix, the degree matrix, and the combinatorial graph Laplacian, respectively. 
Denote by $\mathcal{S}$ a subset of vertices and $\mathcal{S}^c$ its complement.
Then, the cut between $\mathcal{S}$ and $\mathcal{S}^c$ is defined as the sum of weights of edges across them whereas the volume of $\mathcal{S}$ is defined as the sum of weighted degrees of vertices in $\mathcal{S}$. 
More formally, we have
\begin{align*}
\mathrm{cut}(\mathcal{S}, \mathcal{S}^c) =\!\!\! \sum_{u\in\mathcal{S},v\in\mathcal{S}^c} \!\! \mathbf{A}_g(u,v), \,\,
\mathrm{vol}(\mathcal{S})=\sum_{u\in\mathcal{S}}\mathbf{D}_g(u,u).
\end{align*}
One well-known measure for evaluating the partition is normalized cut (Ncut) \cite{shi2000normalized} defined as
\begin{align*}
\mathrm{Ncut}(\mathcal{S}_1,\cdots,\mathcal{S}_k) = \sum_{i=1}^k \frac{\mathrm{cut}(\mathcal{S}_i,\mathcal{S}_i^c)}{\mathrm{vol}(\mathcal{S}_i)}.
\end{align*}
If we define an $N\times k$ matrix $\mathbf{Q}$ whose entries are 
\begin{align}\label{E:indicator}
Q_{vi} = \begin{cases} 
1/\sqrt{\mathrm{vol}(\mathcal{S}_i)}&\mbox{if } v\in\mathcal{S}_i, \\
0 &\mbox{otherwise},
\end{cases} 
\end{align} 
then it can be shown that $\mathrm{Ncut}(\mathcal{S}_1,\cdots,\mathcal{S}_k) = \mathrm{Tr}(\mathbf{Q}^{\top}\mathbf{L}_g\mathbf{Q})$. 
Thus, we can write the problem of minimizing the Ncut as
\begin{align}\label{E:discrete}
\min_{\mathcal{S}_1,\cdots,\mathcal{S}_k} \mathrm{Tr}(\mathbf{Q}^{\top}\mathbf{L}_g\mathbf{Q})\quad\text{ s.t. } \mathbf{Q}^{\top}\mathbf{D}_g\mathbf{Q} = \mathbf{I}, \,\, \mathbf{Q} \text{ as in } \eqref{E:indicator}. 
\end{align}
The spectral graph partitioning method~\cite{shi2000normalized} relaxes~\eqref{E:discrete} to a continuous optimization problem by ignoring its second constraint.
The solution to the relaxed problem is the $k$ generalized eigenvectors of $\mathbf{L}_g\mathbf{q}_i=\lambda_i\mathbf{D}_g\mathbf{q}_i$ associated with the $k$ smallest eigenvalues. 
Then, $k$-means~\cite{lloyd1982least} can be applied to the rows of $\mathbf{Q}=[\mathbf{q}_1,\cdots,\mathbf{q}_k]$ to obtain the desired clusters $\mathcal{S}_1,\cdots,\mathcal{S}_k$.


\section{The Proposed Hypergraph Co-clustering}\label{sec-proposed}

\subsection{Star expansion and hypergraph Laplacians}\label{ssec-hgl}

We project the hypergraph $\mathcal{H}$ onto a directed graph $\mathcal{G}_s=(\mathcal{V}_s,\mathcal{E}_s)$ via the so-called star expansion, where we replace each hyperedge with a star graph. 
More precisely, we introduce a new vertex for every hyperedge $e\in\mathcal{E}$, thus $\mathcal{V}_s=\mathcal{V}\cup\mathcal{E}$. 
The graph $\mathcal{G}_s$ connects each new vertex representing a hyperedge $e$ with each vertex $v\in e$ through two directed edges (one in each direction) that we weigh differently, as explained next.

We consider a random walk on the hypergraph $\mathcal{H}$ (equivalently, on $\mathcal{G}_s$) in which we walk from a vertex $v$ to a hyperedge $e$ that contains $v$ with probability proportional to $\omega(e)$, and then walk from $e$ to a vertex $u$ contained in $e$ with probability proportional to $\gamma_e(u)$. 
We define two matrices $\mathbf{P}_{\mathcal{V}\to\mathcal{E}}\in\mathbb{R}^{|\mathcal{V}|\times|\mathcal{E}|}$ and $\mathbf{P}_{\mathcal{E}\to\mathcal{V}}\in\mathbb{R}^{|\mathcal{E}|\times|\mathcal{V}|}$ to collect the transition probabilities from $\mathcal{V}$ to $\mathcal{E}$ and from $\mathcal{E}$ to $\mathcal{V}$, respectively.
The corresponding entries are given by $\mathbf{P}_{\mathcal{V}\to\mathcal{E}}(v,e)=W_{ve}/\sum_{e'}W_{ve'}$ and $\mathbf{P}_{\mathcal{E}\to\mathcal{V}}(e,v)=R_{ev}/\sum_{v'}R_{ev'}$. 
Then, the transition probability matrix associated with a random walk on $\mathcal{G}_s$ can be written as 
\begin{align*}
\mathbf{P} = 
\left[\begin{matrix}
\mathbf{0}_{|\mathcal{V}|\times|\mathcal{V}|} & 	 \mathbf{P}_{\mathcal{V}\to\mathcal{E}} \\
\mathbf{P}_{\mathcal{E}\to\mathcal{V}} & \mathbf{0}_{|\mathcal{E}|\times|\mathcal{E}|}
\end{matrix}\right].
\end{align*}
When the hypergraph $\mathcal{H}$ is connected, the graph $\mathcal{G}_s$ is strongly connected, thus the random walk defined by $\mathbf{P}$ is irreducible (every vertex can reach any vertex). 
Moreover, it is periodic since $\mathcal{G}_s$ is bipartite and once we start at a vertex $v$, we can only return to $v$ after even steps.

It is well known that a random walk has a unique stationary distribution if it is irreducible and aperiodic~\cite{chung2005laplacians}. 
To fix the above periodicity problem, we introduce self-loops to $\mathcal{G}_s$ and define a new transition probability matrix $\mathbf{P}_{\alpha}=(1-\alpha)\mathbf{I}+\alpha\mathbf{P}$ where $0<\alpha<1$. Matrix $\mathbf{P}_{\alpha}$ defines a random walk (the so-called lazy random walk) where at each discrete time point we take a step of the original random walk with probability $\alpha$ and stay at the current vertex with probability $1-\alpha$. 
The stationary distribution $\pmb{\pi}$ of the random walk is the all-positive dominant left eigenvector of $\mathbf{P}_{\alpha}$, i.e. $\pmb{\pi}^{\top}\mathbf{P}_{\alpha}=\pmb{\pi}^{\top}$, scaled to satisfy $\|\pmb{\pi}\|_1=1$. Notice that different choices of $\alpha$ lead to the same $\pmb{\pi}$.

\begin{table*}
\caption{Summary of datasets considered.}\label{table-datasets}	
\centering
\begin{tabular}{ccccc}
\hline
Datasets & Subsets & \# documents & \# words & Classes \\
\hline
\multirow{2}*{20 Newsgroups} 
& Dataset 1 & 3,863 & 2,000 & comp.os.ms-windows.misc, rec.autos, sci.crypt, talk.politics.guns \\
& Dataset 2 & 5,663 & 2,000 & alt.atheism, comp.graphics, misc.forsale, rec.sport.hockey, sci.electronics, talk.politics.mideast \\
\hline
\multirow{2}*{RCV1}          
& Dataset 3 & 4,000 & 2,000 & CCAT, ECAT, GCAT, MCAT \\
& Dataset 4 & 8,000 & 2,000 & C15, C18, E31, E41, GCRIM, GDIS, M11, M14 \\
\hline 
\end{tabular}
\end{table*}

Given $\mathbf{P}_{\alpha}$ and $\mathbf{\Phi} = \mathrm{diag}(\pmb{\pi})$, we generalize the directed combinatorial Laplacian $\mathbf{L}$ and the normalized Laplacian $\mathcal{L}$ \cite{chung2005laplacians} to hypergraphs as follows
\begin{align}
\mathbf{L} &= \mathbf{\Phi} - \frac{\mathbf{\Phi}\mathbf{P}_{\alpha} + \mathbf{P}_{\alpha}^{\top}\mathbf{\Phi}}{2}, \label{E:L1}\\
\mathcal{L} &= \mathbf{\Phi}^{-\frac{1}{2}}\mathbf{L}\mathbf{\Phi}^{-\frac{1}{2}} = \mathbf{I} - \frac{\mathbf{\Phi}^{\frac{1}{2}}\mathbf{P}_{\alpha}\mathbf{\Phi}^{-\frac{1}{2}} + \mathbf{\Phi}^{-\frac{1}{2}}\mathbf{P}_{\alpha}^{\top}\mathbf{\Phi}^{\frac{1}{2}}}{2}. \label{E:L2}
\end{align}
It can be readily verified that \eqref{E:L1} and \eqref{E:L2} are equal to the combinatorial and normalized Laplacians of the undirected graph defined by the following weighted adjacency matrix
\begin{align}\label{E:hg_A}
\mathbf{A} = \frac{\mathbf{\Phi}\mathbf{P}_{\alpha} + \mathbf{P}_{\alpha}^{\top}\mathbf{\Phi}}{2},
\end{align}
where $\mathbf{\Phi}=\mathrm{diag}(\mathbf{A1})$ is the corresponding degree matrix.

\subsection{Spectral hypergraph partitioning}

We can leverage the hypergraph Laplacians proposed in Section~\ref{ssec-hgl} to apply spectral \emph{graph} partitioning methods (as introduced in Section~\ref{ssec-sgp}) to \emph{hypergraphs}. 
More precisely, we compute the $k$ generalized eigenvectors $\mathbf{U}=[\mathbf{u}_1,\cdots,\mathbf{u}_k]$ of the generalized eigenproblem $\mathbf{Lu}=\lambda\mathbf{\Phi u}$ associated with the $k$ smallest eigenvalues, and then cluster the rows of $\mathbf{U}$ using $k$-means. 
Note that $\mathbf{Lu}=\lambda\mathbf{\Phi u}$ can be written as $\mathbf{\Phi}^{-\frac{1}{2}}\mathbf{L}\mathbf{\Phi}^{-\frac{1}{2}}(\mathbf{\Phi}^{\frac{1}{2}}\mathbf{u})=\lambda\mathbf{\Phi}^{\frac{1}{2}}\mathbf{u}$, implying that $(\lambda,\mathbf{\Phi}^{\frac{1}{2}}\mathbf{u})$ is an eigenpair of the normalized Laplacian $\mathcal{L}$. Hence, if $\mathbf{v}$ is an eigenvector of $\mathcal{L}$, then $\mathbf{u}=\mathbf{\Phi}^{-\frac{1}{2}}\mathbf{v}$.

Since obtaining eigenvectors can be computationally challenging, we show next how to compute the eigenvectors of $\mathcal{L}$ from a smaller size matrix. 
To do this, let us first rewrite $\mathbf{P}_{\alpha}$ and $\mathbf{\Phi}$ as 
\begin{align*}
\mathbf{P}_{\alpha} = \left[
\begin{matrix}
(1-\alpha)\mathbf{I}_{|\mathcal{V}|} \!&\! 	 \alpha\mathbf{P}_{\mathcal{V}\to\mathcal{E}} \\
\alpha\mathbf{P}_{\mathcal{E}\to\mathcal{V}} \!&\! (1-\alpha)\mathbf{I}_{|\mathcal{E}|}
\end{matrix}
\right], \,\,
\mathbf{\Phi} = \left[
\begin{matrix}
\mathbf{\Phi}_{\mathcal{V}} \!&\! \mathbf{0}_{|\mathcal{V}|\times|\mathcal{E}|} \\
\mathbf{0}_{|\mathcal{E}|\times|\mathcal{V}|} \!&\! \mathbf{\Phi}_{\mathcal{E}}
\end{matrix}
\right].
\end{align*}

\begin{proposition}\label{prop}
Define the following matrix 
\begin{align}
\bar{\mathbf{A}} = \frac{1}{2} \left( \mathbf{\Phi}_{\mathcal{V}}^{\frac{1}{2}}\mathbf{P}_{\mathcal{V}\to\mathcal{E}}\mathbf{\Phi}_{\mathcal{E}}^{-\frac{1}{2}} + \mathbf{\Phi}_{\mathcal{V}}^{-\frac{1}{2}}\mathbf{P}_{\mathcal{E}\to\mathcal{V}}^{\top}\mathbf{\Phi}_{\mathcal{E}}^{\frac{1}{2}} \right),
\end{align}
and denote by $\bar{\mathbf{u}}$ and $\bar{\mathbf{v}}$ the left and right singular vectors of $\bar{\mathbf{A}}$ associated with the singular value $\bar{\lambda}$, respectively. 
Then, the vector $\mathbf{v}=[\bar{\mathbf{u}}^{\top}, \bar{\mathbf{v}}^{\top}]^{\top}$ is the eigenvector of $\mathcal{L}$ associated with the eigenvalue $\lambda=\alpha(1-\bar{\lambda})$.
\end{proposition}
\begin{proof}
Let us rewrite $\mathcal{L}$ as
\begin{align}\label{E:L2_rewrite}
\mathcal{L} = \alpha\left[
\begin{matrix}
	\mathbf{I}_{|\mathcal{V}|} & - \bar{\mathbf{A}} \\
	-\bar{\mathbf{A}}^{\top} & \mathbf{I}_{|\mathcal{E}|}
\end{matrix}
\right].	
\end{align}
Split its eigenvector into two parts $\mathbf{v}=[\mathbf{v}_{\mathcal{V}}^{\top}, \mathbf{v}_{\mathcal{E}}^{\top}]^{\top}$ where $\mathbf{v}_{\mathcal{V}}$ and $\mathbf{v}_{\mathcal{E}}$ respectively have length $|\mathcal{V}|$ and $|\mathcal{E}|$. Then we have 
\begin{align*}
\alpha\left[\begin{matrix}
	\mathbf{I}_{|\mathcal{V}|} & - \bar{\mathbf{A}} \\
	-\bar{\mathbf{A}}^{\top} & \mathbf{I}_{|\mathcal{E}|}
\end{matrix}\right] 
\left[
\begin{matrix}
\mathbf{v}_{\mathcal{V}} \\
\mathbf{v}_{\mathcal{E}}
\end{matrix}
\right] = \lambda
\left[
\begin{matrix}
\mathbf{v}_{\mathcal{V}} \\
\mathbf{v}_{\mathcal{E}}
\end{matrix}
\right],
\end{align*}
and it follows that 
\begin{align*}
\bar{\mathbf{A}}\mathbf{v}_{\mathcal{E}} =(1-\alpha^{-1}\lambda)\mathbf{v}_{\mathcal{V}},
\qquad \bar{\mathbf{A}}^{\top}\mathbf{v}_{\mathcal{V}}=(1-\alpha^{-1}\lambda)\mathbf{v}_{\mathcal{E}}.
\end{align*}
When $1-\alpha^{-1}\lambda>0$, i.e. $\lambda<\alpha$, $\mathbf{v}_{\mathcal{V}}$ and $\mathbf{v}_{\mathcal{E}}$ are respectively the left and right singular vectors of $\bar{\mathbf{A}}$ and $1-\alpha^{-1}\lambda$ is the corresponding singular value.
\end{proof}

Based on Proposition~\ref{prop}, our proposed spectral hypergraph co-clustering algorithm is given by the following steps:\\
1) Compute the $k$ left and right singular vectors of $\bar{\mathbf{A}}$ associated with the $k$ largest singular values, denoted by $\bar{\mathbf{U}}\in\mathbb{R}^{|\mathcal{V}|\times k}$ and $\bar{\mathbf{V}}^{|\mathcal{E}|\times k}$, respectively.\\
2) Leverage Proposition~\ref{prop} to form $\mathbf{U} = [\mathbf{\Phi}_{\mathcal{V}}^{-\frac{1}{2}}\bar{\mathbf{U}}; \mathbf{\Phi}_{\mathcal{E}}^{-\frac{1}{2}}\bar{\mathbf{V}}]$. \\
3) (Optional) Normalize the rows of $\mathbf{U}$ to have unit norm. \\
4) Apply $k$-means to the rows of $\mathbf{U}$ (or its normalized version).

The optional normalization step above is inspired by the spectral partitioning algorithm proposed in~\cite{ng2002spectral}.
In our next section, we denote the variant of our algorithm without normalization as s-spec-1 whereas the one that implements the third step above is denoted as s-spec-2.

\vspace{1mm}
\noindent 
\emph{How to choose parameter $\alpha$?} From Proposition~\ref{prop} and \eqref{E:L2_rewrite} we can see that the choice of $\alpha$ affects the eigenvalues of $\mathcal{L}$ but does not change its eigenvectors (or their order).
Hence, the proposed spectral clustering method is independent of $\alpha$.


\begin{figure*}
\centering
\includegraphics[scale=0.4]{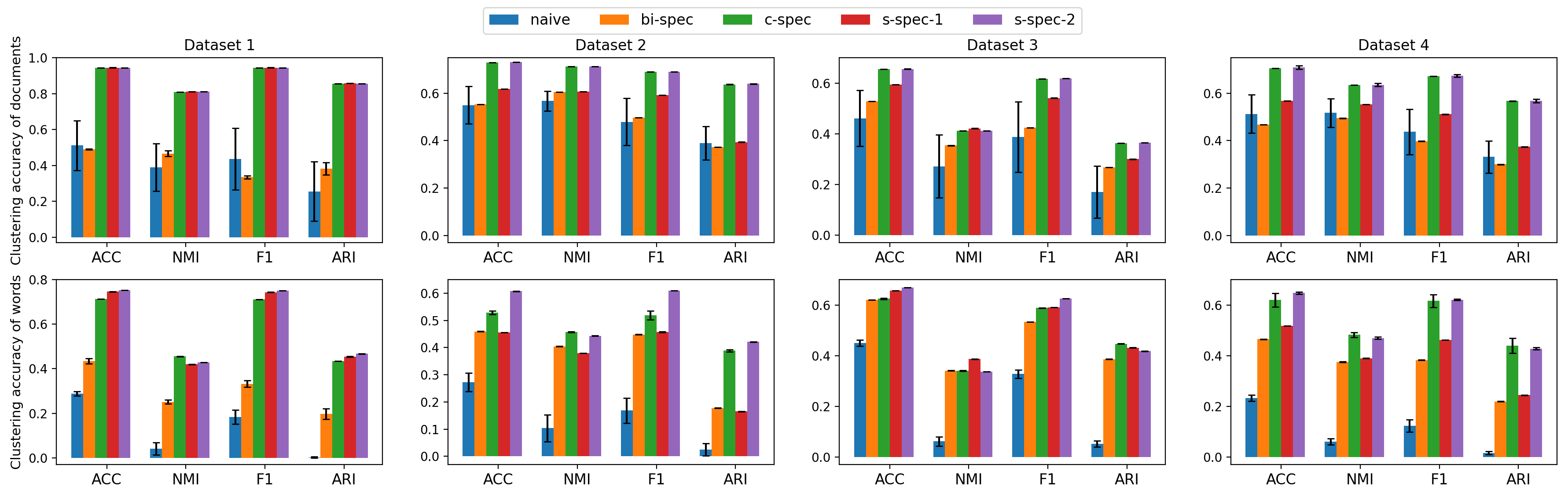}	
\caption{Performance comparison of clustering algorithms. The two rows respectively show the clustering accuracy of documents and words. Each column corresponds to one dataset.}
\label{fig-acc}
\vspace{-0.8em}
\end{figure*}

\section{Experiments}


In this section, we evaluate the performance of the proposed methods via numerical experiments.\footnote{The code needed to replicate the numerical experiments presented in this paper can be found at \url{https://github.com/yuzhu2019/hypergraph_cocluster}.}
We consider two widely used real-world text datasets: 20 Newsgroups\footnote{\url{http://qwone.com/~jason/20Newsgroups/}} and Reuters Corpus Volume 1 (RCV1) \cite{lewis2004rcv1}. 
Both of them contain documents in different categories. 
We extract two subsets of documents from each of them to build datasets of different levels of difficulty (datasets 1 and 3 are easier than datasets 2 and 4; see Table~\ref{table-datasets}). 
We consider the $2,\!000$ most frequent words in the corpus after removing stop words and words appearing in $>20\%$ and $<0.2\%$ of the documents. 


To model text datasets using hypergraphs with EDVWs, we follow the procedure in~\cite{hayashi2020}. 
More precisely, we consider documents as vertices and words as hyperedges.
A document (vertex) belongs to a word (hyperedge) if the word appears in the document. 
The EDVWs (the entries in $\mathbf{R}$) are taken as the corresponding tf-idf (term frequency–inverse document frequency) values, which reflect how relevant a word is to a document in a collection of documents. 
The weight associated with a hyperedge is computed as the standard deviation of the entries in the corresponding row of $\mathbf{R}$.   


We compare the proposed methods (s-spec-1 and s-spec-2) with the following three methods. 
(i)~The naive method (naive): We run $k$-means on the columns and the rows of the tf-idf matrix $\mathbf{R}$ to cluster documents and words, respectively.  
(ii)~Bipartite spectral graph partitioning (bi-spec)~\cite{dhillon2001co}: The dataset is modeled as an (undirected) bipartite graph between documents and words, then a spectral graph partitioning algorithm is applied; see Section~\ref{ssec-sgp}. 
~
(iii)~Clique expansion (c-spec, Algorithm 1 in \cite{hayashi2020}): This method projects the hypergraph with EDVWs onto a proxy graph via the so-called clique expansion, then applies a spectral graph partitioning algorithm. 
We consider it as the state-of-the-art method.
Since c-spec can only cluster the vertices (and not the hyperedges), we build a hypergraph as mentioned above to cluster documents and then we construct another hypergraph in which we take words as vertices and documents as hyperedges to cluster words.  
Notice that of the  above mentioned methods only the proposed methods (s-spec-1 and s-spec-2) and bi-spec can co-cluster documents and words.

\begin{figure}
\centering
\includegraphics[scale=0.5]{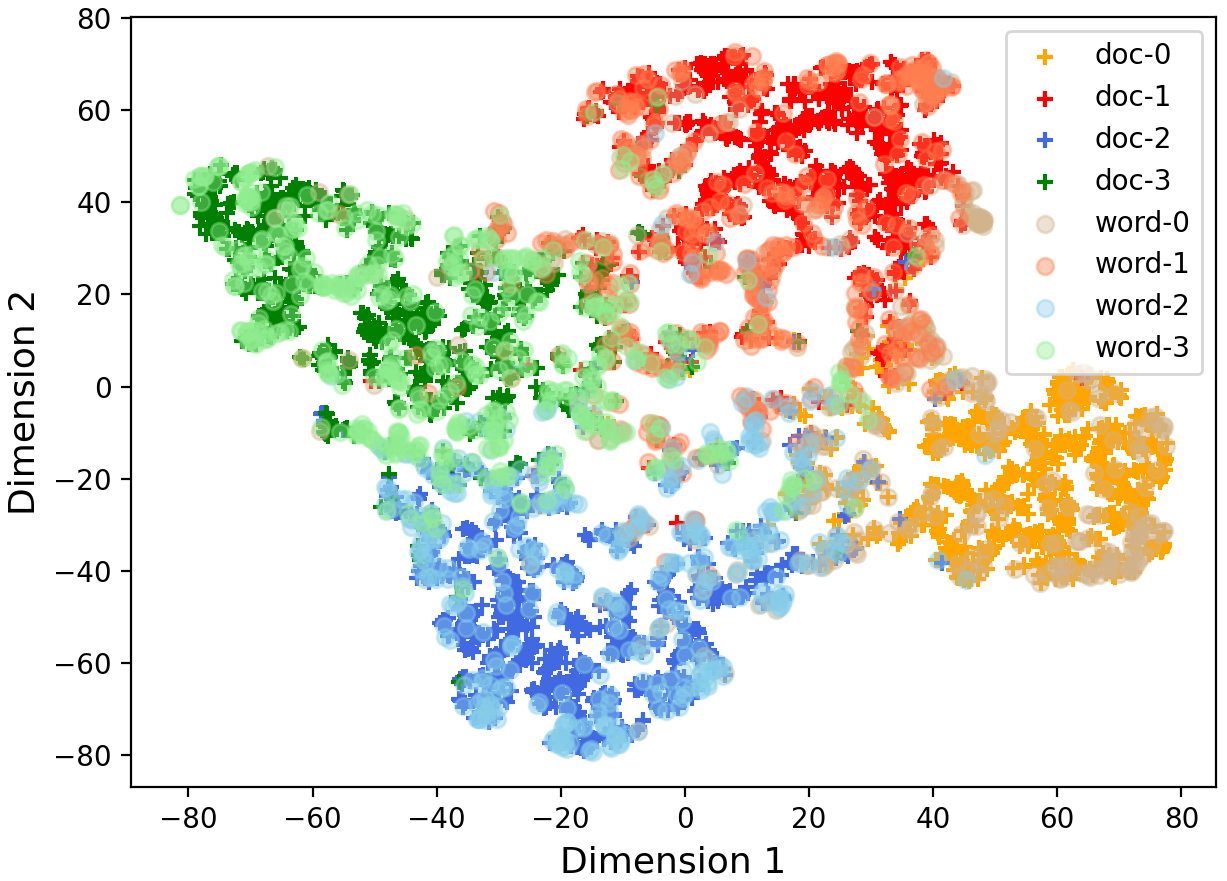}
\caption{The 2D t-SNE plot of document and word embeddings learned by s-spec-2 in Dataset 1. doc-$i$ and word-$i$ indicate documents and words from the four classes in the dataset.}
\vspace{-0.8em}
\label{fig-tsne}	
\end{figure}


To evaluate the clustering performance, we consider four metrics, namely, clustering accuracy score (ACC), normalized mutual information (NMI), weighted F1 score (F1), and adjusted Rand index (ARI)~\cite{emmons2016analysis}. 
For all of them, a larger value indicates a better performance. 
Notice that there are no ground-truth classes for words. 
Hence, following~\cite{ding2006orthogonal}, we consider the class conditional word distribution. 
More precisely, we compute the aggregate word distribution for each document class, then for every word we assign it to the class in which it has the highest probability in the aggregate distribution. 
We regard this assignment as the ground truth for performance evaluation. 

The numerical results (averaged over $10$ runs of $k$-means) are shown in Fig.~\ref{fig-acc}. 
We first notice that, of the proposed methods, s-spec-2 usually performs better than s-spec-1. 
This is in line with~\cite{ng2002spectral}, where it was observed that the lack of a normalization step (as in our s-spec-1) might lead to performance decays when the connectivity within each cluster varies substantially across clusters.
It can also be seen that the proposed methods and c-spec tend to work better than the naive method and the classical bipartite spectral graph partitioning method. 
This underscores the value of the hypergraph model considered. 
Importantly, s-spec-2 achieves similar clustering accuracy as the state-of-the-art c-spec for documents but tends to perform better in clustering words. 
Moreover, the proposed methods achieve small standard deviations, indicating their robustness to different centroid initializations in $k$-means.
	
\begin{figure}
\centering
\includegraphics[scale=0.47]{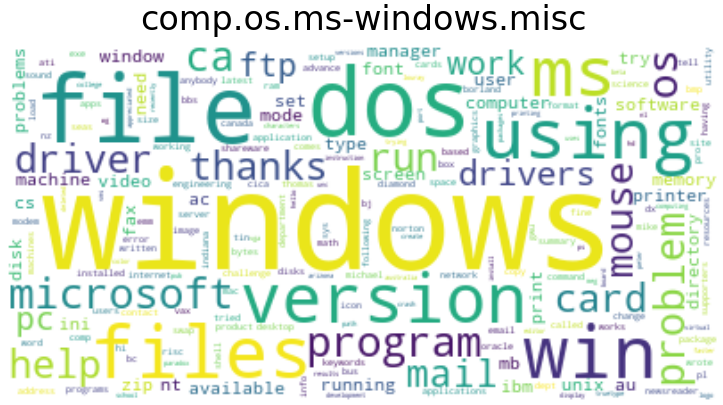}	
\includegraphics[scale=0.47]{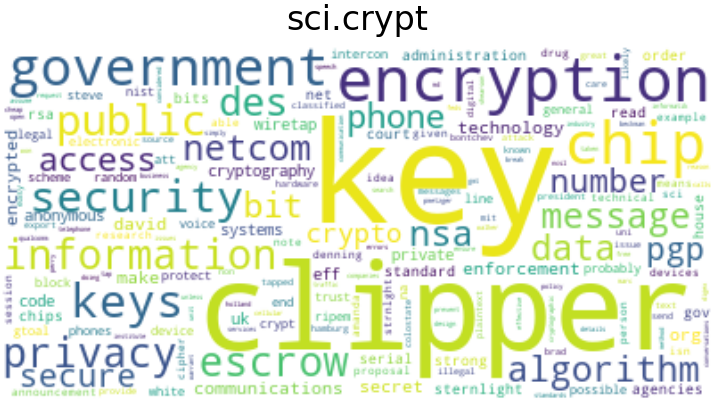}	
\caption{Word clouds for words predicted in the classes `comp.os.ms-windows.misc' and `sci.crypt'.}
\vspace{-0.8em}
\label{fig-wc}
\end{figure}	
	
Having showed the superiority in performance of s-spec-2, we now present visualizations of its application to Dataset 1 to further illustrate its effectiveness. 
In Fig.~\ref{fig-tsne}, we depict the embeddings of documents and words obtained by s-spec-2 by mapping them to a 2D space using t-SNE \cite{maaten2008visualizing}. 
We can see that documents and words in the same class appear to form groups. 
In Fig.~\ref{fig-wc}, we plot the word clouds\footnote{\url{https://github.com/amueller/word_cloud}} for the words predicted in the classes `comp.os.ms-windows.misc' (Microsoft Windows operating system) and `sci.crypt' (cryptography). 
The size of a word is determined by its frequency in the documents predicted in the same class, thus is able to reveal its importance in the class. 
We can see that the top words (such as windows, file, dos, ms in `comp.os.ms-windows.misc') align well with our intuitive understanding of the class topics.


\section{Conclusions}

We developed valid Laplacian matrices for hypergraphs with EDVWs, based on which we proposed spectral partitioning algorithms for co-clustering vertices and hyperedges. 
Through real-world text mining applications, we showcased the value of considering hypergraph models and demonstrated the effectiveness of our proposed methods.
Future research avenues include: 
(i) Developing alternative co-clustering methods where we replace the spectral clustering step by non-negative matrix tri-factorization algorithms~\cite{ding2006orthogonal, shang2012graph, wang2011fast} of matrices related to the hypergraph Laplacians. 
(ii) Generalizing additional existing digraph Laplacians \cite{li2010random, cucuringu2020hermitian} to the hypergraph case.
(iii) Study the use of the hypergraph model with EDVWs in other network analysis tasks such hypergraph alignment~\cite{zass2008probabilistic, tan2014mapping, mohammadi2016triangular}. 
Related to this last point, the fact that our proposed methods embed vertices and hyperedges in the same vector space (as shown in Fig.~\ref{fig-tsne}) facilitates the development of embedding-based hypergraph alignment algorithms~\cite{heimann2018regal}.

{\footnotesize\bibliography{star_expansion}}
\bibliographystyle{IEEEbib}

\end{document}